\documentclass[a4paper]{amsart}

\usepackage{amssymb,amsmath,graphicx,epstopdf}
\usepackage{enumitem}
\usepackage{framed}
\usepackage[all]{xy}
\usepackage{tikz}
\usepackage{algorithm}
\usepackage{algpseudocode}

\usetikzlibrary{positioning}

\theoremstyle{plain}
\newtheorem{theorem}{Theorem}[section]
\newtheorem{definition}[theorem]{Definition}
\newtheorem{lemma}[theorem]{Lemma}

\newtheorem{proposition}[theorem]{Proposition}

\newtheorem*{CSPdichotomy}{CSP dichotomy conjecture}

\newtheorem*{algdichotomy}{Algebraic CSP dichotomy conjecture}

\newtheorem*{openproblem}{Problem}

\newcommand{\CSP}{\operatorname{CSP}}

\begin{document}

\title[CSP Dichotomy via consistency checks]{Constraint Satisfaction Problem Dichotomy for Finite Templates: a proof via consistency checks}

\author{Dejan Deli\'c}
\address{Department of Mathematics, Ryerson University,  Canada}
\email{ddelic@ryerson.ca}

\thanks{The author gratefully acknowledges support by the Natural Sciences and
Engineering Research Council of Canada in the form of a Discovery Grant.}

\begin{abstract} 
One of the central problems in the study of parametrized constraint satisfaction problems is the Dichotomy Conjecture by T. Feder and M. Vardi stating that the constraint satisfaction problem (CSP) over a fixed, finite constraint language is either solvable in polynomial time or \textsc{NP}-complete. The conjecture was verified in certain special cases (domains with a relatively small number of elements, constraint languages containing all unary relations, etc.)  In this article, we present a proof of the Dichotomy Conjecture via  local consistency and a new consistency notion, the AF-consistency checks. In fact, we show that, for every Taylor domain, which is $(2\lceil\frac{K}{2}\rceil,3\lceil\frac{K}{2}\rceil)$-consistent, where $K$ is the largest arity of a relation in the constraint language, we can define polynomially many proper subinstances such that, the original instance of a CSP is solvable if, and only if, the problem has a solution in one of those subinstances and define the AF-consistency using these subinstance . Finally, the solution is constructed, via a sequence of reductions using absorption and AF-consistency, using the notion of Singleton Linear Arc Consistency (SLAC), as introduced in \cite{Kozik2016}.
\end{abstract}

\maketitle

\section*{Introduction}
\noindent One of the fundamental problems in constraint programming and, more widely, in the field of artificial intelligence, is the problem of understanding the computational complexity of constraint satisfaction problems (CSPs, for short). The problem, in its full generality, is \textsc{NP}-complete but may of its subclasses are tractable and can be solved efficiently using well-established algorithms. One approach to studying the complexity of CSPs is to restrict the instances by allowing a fixed set of constraint relations, which is generally referred to in the literature as a \emph{constraint language} or, a \emph{fixed template}  (\cite{b-j-k}). This particular approach has proved to be very fruitful and has been the driving force in the study of the complexity of constraint satisfaction in the past 15 years or so.

The fixed template approach can be traced back to the 1970s and the work of U. Montanari (\cite{Montanari}) and T. Schaefer (\cite{sch}). The latter work resulted in the first general result in the area, Schaefer's dichotomy for Boolean CSPs. Schaefer proved that CSPs arising from constraint languages over 2-element domains are either solvable in polynomial time or \textsc{NP}-complete. Another landmark result was the dichotomy for finite simple graph templates by P. Hell and J. Ne\v{s}et\v{r}il (\cite{helnes:1}). Their result proves that if a fixed template is a finite simple graph (i.e. the domain is the set of all vertices and the edge relation is the only constraint), then the associated CSP is either solvable in polynomial time or \textsc{NP}-complete.

These seminal results concerning the dichotomy of CSPs over rather specific types of finite templates gave rise to a more general project of classifying the complexity of \emph{all} fixed template CSPs. The main conjecture in the field and the subject of this article is the so-called \emph{CSP Dichotomy Conjecture}, formulated by T. Feder and M. Vardi (\cite{fedvar}). The dichotomy conjecture~\cite{fedvar} can be stated as
follows:

\begin{CSPdichotomy}
  Let $\mathbf{A}$ be a finite relational structure.  
  Then\linebreak
  $\CSP(\mathbf{A})$ is solvable in polynomial time or \textsc{NP}-complete.
\end{CSPdichotomy}

In this article, we provide a proof of this conjecture. 

A big advance in the study of the computational complexity of CSPs was made in the work of P. Jeavons, D. Cohen, and M. Gyssens (\cite{JCG97}) and later extended by the work of A. Bulatov, P. Jeavons, and A. Krokhin (\cite{b-j-k}). The aforementioned articles studied an algebraic connection between fixed template CSPs and their complexity. Namely, one can associate with every finite domain constraint language $\mathbf{A}$ a finite algebraic structure, its \emph{algebra of polymorphisms}. The properties of the algebraic structure obtained in this way directly influence the computational complexity of the constraint language. In particular, if the fixed constraint language is a relational core, if its algebra of polymorphisms does not satisfy a particular equational property, referred to as a \emph{Taylor property} in universal algebra, then the CSP associated with the constraint language is \textsc{NP}-complete.

In \cite{b-j-k}, the authors conjectured that all constraint languages over finite domains whose algebras of polymorphisms are Taylor give rise to CSPs solvable in polynomial time. This conjecture is often referred to as the \emph{Algebraic CSP Dichotomy Conjecture}

\begin{algdichotomy}
  Let $\mathbf{A}$ be a finite relational structure that is a core. If
  the algebra of polymorphisms of 
  $\mathbf{A}$ is Taylor, then $\CSP(\mathbf{A})$ is
  solvable in polynomial time, otherwise $\CSP(\mathbf{A})$ is
  \textsc{NP}-complete.
\end{algdichotomy}

The algebraic approach has subsequently yielded a number of important results. Among others, A. Bulatov \cite{bul3} extended
Schaefer's~\cite{sch} result on 2-element domains to prove the CSP  
dichotomy conjecture for 3-element domains, with further results by other authors extending it to the domains of cardinality $\leq 7$. L. Barto, Kozik and
T. Niven~\cite{b-k-n} extended P. Hell and J. Ne{\v s}et{\v r}il's
result~\cite{helnes:1} on simple graphs to constraint languages
consisting of a finite digraph with no sources and no sinks. Barto and
Kozik~\cite{b-k2} gave a complete algebraic description of the
constraint languages over finite domains that are solvable by local
consistency methods (these problems are said to be of \emph{bounded
  width}) and as a consequence it is decidable to determine whether a
constraint language can be solved by such methods.

\subsection*{Organization of the paper}

In Section 1 we present the basic concepts and definitions related to constraint satisfaction problems, as well as the tools from universal algebra which will be used extensively in the paper. In Section 2, we describe the local consistency notions which are used in the proof of the main result. Section 3 outlines the reduction to the binary case, i.e. the syntactically simple instances and their combined vertex-edge instances. Section 4 contains a proof of the Dichotomy Conjecture using the tools introduced earlier in the paper, along with the new consistency concept, related to solvability of cyclic constraint satisfaction problems, while Section 5 states the algorithm which solves CSPs with finite Taylor constraint languages.

\section{Preliminaries}

\subsection{Constraint Satisfaction Problems}
The central concept of this paper is the one of a non-uniform Constraint Satisfaction Problem:

\begin{definition} An \emph{instance} of the CSP is a triple $\mathcal{I}=(V,A,\mathcal{C})$, where $V=\{x_1,\ldots,x_n\}$ is a finite set of \emph{variables}, $A$ is a finite domain for the variables in $V$, and $\mathcal{C}$ is a finite set of \emph{constraints} of the form $C=(S,R_S)$, where $S$, the \emph{scope} of the constraint, is a $k$-tuple of variables $(x_{i_1},\ldots,x_{i_k})\in V^k$ and $R_S$ is a $k$-ary relation $R_S\subseteq A^k$, called the \emph{constraint relation} of $C$.

A \emph{solution} for the instance $\mathcal{I}$ is any assignment $f:V\rightarrow A$, such that, for every constraint $C=(S,R_S)$ in $\mathcal{C}$, $f(S)\in R_S$.
\end{definition}

A relational structure $\mathbf{A}=(A, \Gamma)$, defined over the domain $A$ of the instance $\mathcal{I}$, where $\Gamma$ is a finite set of relations on $A$, is a \emph{constraint language}, with the relations from $\Gamma$ forming the signature of $\mathbf{A}$. An instance of $\CSP (\mathbf{A})$ is an instance of the CSP such that all constraint relations belong to $\mathbf{A}$. 

We can now formulate the constraint satisfaction problem $\CSP (\mathbf{A})$ as the following decision problem:


We will not dwell on the particular issue of how the relations of $\Gamma$ are represented as a part of the input. It suffices to note that all standard ways of representing the constraint relations in the literature lead to log-space equivalent decision problems.

\subsection{Basic Algebraic Tools}\label{sec:Algebra}
\noindent In this section, we introduce concepts from universal algebra which will be used in the remainder of the paper. For more exhaustive introduction to universal algebra and its applications, see \cite{Burris1981} or \cite{bergman}.

An \emph{algebra} is an ordered pair $\mathbb{A}=(A, F)$, where $A$ is a nonempty set, the \emph{universe} of $\mathbb{A}$, while $F$ is the set of \emph{basic operations} of $\mathbb{A}$, consisting of functions of arbitrary, but finite, arities on $A$. The list of function symbols and their arities is the \emph{signature} of $\mathbb{A}$. 

A \emph{subuniverse} of the algebra $\mathbb{A}$ is a nonempty subset $B\subseteq A$ closed under all operations of $\mathbb{A}$. If $B$ is a subuniverse of $\mathbb{A}$, by restricting all operations of $\mathbb{A}$ to $B$, such a subuniverse is a \emph{subalgebra} of $\mathbb{A}$, which we denote $\mathbb{B}\leq \mathbb{A}$.

If $\mathbb{A}_i$ is an indexed family of algebras of the same signature, the product $\prod_i \mathbb{A}_i$ of the family is the algebra whose universe is the Cartesian products of their universes $\prod_i A_i$ endowed with the basic operations which are coordinatewise products of the corresponding operations in $\mathbb{A}_i$. If $\mathbb{A}$ is an algebra, its $n$-th Cartesian power will be denoted $\mathbb{A}^n$.

An equivalence relation $\alpha$ on the universe $A$ of an algebra $\mathbb{A}$ is a \emph{congruence} of $\mathbb{A}$, if $\alpha \leq \mathbb{A}^2$, i.e. if $\alpha$ is preserved by all basic operations of $\mathbb{A}$. In that case, one can define the algebra $\mathbb{A}/\alpha$, the \emph{quotient of} $\mathbb{A}$ \emph{by} $ \alpha$, with the universe consisting of all equivalence classes (cosets) in $A/\alpha$ and whose basic operations are induced by the basic operations of $\mathbb{A}$. The $\alpha$-congruence class containing $a\in A$ will be denoted $a/\alpha$.

An algebra $\mathbb{A}$ is said to be \emph{simple} if its only congruences are the trivial, diagonal relation $0_\mathbb{A}=\{(a,a)\, \vert \, a\in A\}$ and the full relation $1_\mathbb{A}=\{ (a,b)\, \vert\, a,b\in A\}$. It is a well-known fact (see e.g. \cite{Burris1981}) that the congruences of $\mathbb{A}$ form a lattice $Con(\mathbb{A})$; namely, for any $\alpha,\beta\in Con(\mathbb{A})$, $\alpha\wedge\beta$ is the intersection of $\alpha$ and $\beta$, while $\alpha\vee \beta$ is the smallest congruence containing both $\alpha$ and $\beta$.

Any subalgebra of a Cartesian product of algebras $\mathbb{A}\leq \prod_i \mathbb{A}_{i\in I}$ is equipped with a family of congruences arising from projections on the product coordinates. We denote $\pi_i$ the congruence obtained by identifying the tuples in $A$ which have the same value in the $i$-th coordinate. Given any $J\subseteq I$, we can define a subalgebra of $\mathbb{A}$, $proj_J(\mathbb{A})$, which consists of the projections of all tuples in $A$ to the coordinates from $J$. If $\mathbb{A}\leq \prod_{i\in I} \mathbb{A}_i$ is such that $proj_i (\mathbb{A})=\mathbb{A}_i$, for every $i\in I$, we say that $\mathbb{A}$ is a \emph{subdirect product} and denote this fact $\mathbb{A}\leq_{sp} \prod_{i\in I} \mathbb{A}_i$.

If $\mathbb{A}$ and $\mathbb{B}$ are two algebras of the same signature, a mapping from $A$ to $B$ which preserves all basic operations is a \emph{homomorphism}. An \emph{isomorphism} is a bijective homomorphism between two algebras of the same signature.

Given an algebra $\mathbb{A}$, a \emph{term} is a syntactical object describing a composition of basic operations of $\mathbb{A}$. A \emph{term operation} $t^\mathbb{A}$ of $\mathbb{A}$ is the interpretation of the syntactical term $t(x_1,\ldots,x_m)$ as an $m$-ary operation on $A$, according to the formation tree of $t$.

A \emph{variety} is a class of algebras of the same signature, which is closed under the class operators of taking products, subalgebras, and homomorphic images (or, equivalently, under the formation of quotients by congruence relations.) The variety $\mathcal{V}(\mathbb{A})$ generated by the algebra $\mathbb{A}$ is the smallest variety containing $\mathbb{A}$. Birkhoff's theorem states (see \cite{Burris1981}) states that every variety is an equational class; that is, every variety $\mathcal{V}$ is uniquely determined by a set of identities (equalities of terms) $s\approx t$ so that $\mathbb{A}\in\mathcal{V}$ if and only if $\mathbb{A}\models s\approx t$, for every identity $s\approx t$ in the set.

\subsection{Homomorphisms, cores and polymorphisms}\label{sec:homs}
\noindent In order to be able to fully utilize the power of the algebraic approach to studying the complexity of CSPs, in this subsection we outline the connection between the constraint satisfaction problems on finite relational templates and their algebraic parametrization. We begin with the notion of a relational structure
homomorphism.

An \emph{$n$-ary operation} on a set $A$ is simply a mapping
$f:A^n\rightarrow A$; the number $n$ is the \emph{arity} of $f$.  Let
$f$ be an $n$-ary operation on $A$ and let $k>0$. We write $f^{(k)}$
to denote the $n$-ary operation obtained by applying $f$ coordinatewise on
$A^k$. That is, we define the $n$-ary operation $f^{(k)}$ on $A^k$ by
\[
f^{(k)}(\mathbf a^1,\dots,\mathbf
a^n)=(f(a^1_1,\dots,a^n_1),\dots,f(a^1_k,\dots,a^n_k)),
\]
for $\mathbf a^1,\dots, \mathbf a^n\in A^k$.

\begin{definition}\label{def:hom}
  Let $\mathbf A$ and $\mathbf B$ be relational structures in the same
  signature~$\Gamma$. A \emph{homomorphism} from $\mathbf A$ to
  $\mathbf B$ is a mapping $\varphi$ from $A$ to $B$ such that for
  each $k$-ary relation symbol $R$ in $\Gamma$ and each $k$-tuple
  $\mathbf{a}\in A^k$, if $\mathbf{a}\in R^\mathbf A$, then
  $\varphi^{(k)}(\mathbf{a})\in R^\mathbf B$.
\end{definition}

We write $\varphi:\mathbf A\to\mathbf B$ to mean that $\varphi$ is a
homomorphism from $\mathbf A$ to $\mathbf B$, and $\mathbf A\to\mathbf
B$ to mean that there exists a homomorphism from $\mathbf A$ to
$\mathbf B$. 

An \emph{isomorphism} is a bijective homomorphism $\varphi$ such that
$\varphi^{-1}$ is also a homomorphism. A homomorphism $\mathbf A\to\mathbf
A$ is called an \emph{endomorphism}. 

A finite relational structure $\mathbf A'$ is a \emph{core} if every
endomorphism $\mathbf A'\to\mathbf A'$ is surjective. For every $\mathbf A$ there exists a relational
structure $\mathbf A'$ such that $\mathbf A\to\mathbf A'$ and $\mathbf
A'\to\mathbf A$ and $\mathbf A'$ is of minimum size with respect to these
properties; that structure $\mathbf A'$ is called the \emph{core of
  $\mathbf A$}. The core of $\mathbf A$ is unique (up to isomorphism)
and $\CSP(\mathbf A)$ and $\CSP(\mathbf A')$ are the same decision
problems. Equivalently, the core of $\mathbf A$ can be defined as an induced substructure of minimum size that $\mathbf A$ retracts onto. (See~\cite{helnes} for details on cores for graphs, cores for
relational structures are a natural generalization.)

The notion of \emph{polymorphism} is central in the 
so-called 
algebraic approach to the $\CSP$. Polymorphisms are a natural
generalization of endomorphisms to higher arity operations.

\begin{definition}
  Given an $\Gamma$-structure $\mathbf{A}$, an $n$-ary
  \emph{polymorphism} of $\mathbf{A}$ is an $n$-ary operation $f$ on
  $A$ such that $f$ preserves the relations of $\mathbf A$. That is,
  if $\mathbf{a}^1,\dots,\mathbf{a}^n\in R$, for some $k$-ary relation
  $R$ in $\Gamma$, then $f^{(k)}(\mathbf a^1,\dots,\mathbf
  a^n)\in R$.  
\end{definition}
Thus, an endomorphism is a unary polymorphism. Polymorphisms satisfying certain identities has been used extensively 
in the algebraic study of CSPs.

Furthermore, if a relational structure $\mathbf{A}$ is a core, one can construct a structure $\mathbf{A}'$ from $\mathbf{A}$ by adding, for each element $a\in A$, a unary constraint relation $\{a\}$. This enables us to further restrict the algebra of polymorphisms associated with the template; namely, if $f(x_1,\ldots,x_m)$ is an $m$-ary polymorphism of $\mathbf{A}'$, it is easy to see that
$f(a,a,\ldots,a)=a,$
for all $a\in A$. In addition to this, the constraint satisfaction problems with the templates $\mathbf{A}$ and $\mathbf{A}'$ are log-space equivalent. Therefore, we may assume that the algebra of polymorphisms associated to any CSP under consideration is \emph{idempotent}; i.e. all its basic operations $f$ satisfy the identity
$$f(x,x,\ldots,x)\approx x.$$

\subsection{Taylor Algebras}

One of the great accomplishments of the algebraic approach, even at its early stages, was strong evidence that the algebraic parametrization $\mathbb{A}$ fully determines the computational complexity of the associated problem $\CSP \mathbf{A})$. The algebraic version of the Dichotomy Conjecture speculates that there is a strict dividing line between tractable and \textsc{NP}-complete problems: if there is a two-element quotient algebra of a subalgebra of $\mathbb{A}$ all of whose operations are projections, then $\CSP (\mathbf{A})$ is \textsc{NP}-complete; otherwise, $\CSP (\mathbf{A})$ is solvable in polynomial time. The hardness part of the Algebraic Dichotomy Conjecture is known to be true:

\begin{theorem}(\cite{b-j-k}) Let $\mathbf{A}$ be a finite relational template which contains all constant unary relations and let $\mathbb{A}$ be its algebra of polymorphisms. If $\mathbb{A}$ contains a subalgebra with a two-element quotient algebra whose only operations are projections, then $\CSP (\mathbf{A})$ is \textsc{NP}-complete. In fact, \textsc{3-SAT} can be polynomially reduced to it.
\end{theorem}

The algebras which fail the assumptions of the theorem and which are conjectured to give rise to tractable problems are called \emph{Taylor algebras}. Therefore, an algebra $\mathbb{A}$ is Taylor if no subalgebra of $\mathbb{A}$ has a two-element quotient algebra whose only operations are projections. 

Taylor algebras can be characterized in a variety of equational ways. For our purposes, besides the specific assumption that an algebra be Taylor, such characterizations will be irrelevant. However, we prefer to state the following theorem which characterizes Taylor algebras in terms of equational logic.

\begin{theorem}(\cite{maroti-mckenzie}) Let $\mathbb{A}$ be a finite idempotent algebra. Then, the following are equivalent:

\begin{itemize}
\item $\mathbb{A}$ is a Taylor algebra.
\item $\mathbb{A}$ has a $k$-ary \emph{weak near-unanimity} operation, for some $k\geq 3$; i.e. a $k$-ary operation satisfying
$$f(x,x,\ldots,,x,y)\approx f(x,x,\ldots,x,y,x)\approx\ldots\approx f(y,x,x,\ldots,x).$$
\end{itemize}
\end{theorem}
 
\subsection{Absorption}\label{absorption}

One of the key notions which has emerged in recent years as an important tool in the algebraic approach to the study of CSPs with finite templates is the one of absorption. It has played a crucial role in the proof of the Bounded Width Conjecture and its refinements (see \cite{b-k1} , \cite{b-k2},  \cite{Kozik2016}) but its primary strength is in its applicability outside the context of congruence meet-semidistributivity.

If $\mathbb{A}$ and $\mathbb{B}$ are idempotent algebras such that $\mathbb{B}\leq \mathbb{A}$, we say that $\mathbb{B}$ \emph{absorbs} $\mathbb{A}$ and write it as $\mathbb{B}\unlhd \mathbb{A}$ if there exists a term $t$
such that
$$t(B,B,\ldots, B,A,B,\ldots, B)\subseteq B,$$
regardless of the placement of $A$ in the list of variables of the term.

A direct consequence of the definition is the following fact: if $\mathbb{A},\mathbb{B},\mathbb{A}'$ and $\mathbb{B}'$ are algebras of the same signature such that $\mathbb{B}\unlhd\mathbb{A}$ and $\mathbb{B}'\unlhd\mathbb{A}'$, then both absorptions can be witnessed by the same term. 

Subdirect products of a pair of algebras give rise to pairs of congruences which will be used in the course of the paper in order to prove the so-called ``rectangulation" properties of powers of simple algebras.  

\begin{proposition}\label{absorb} Let $\mathbb{R}\leq_{sp} \mathbb{A}\times \mathbb{B}$.
\begin{enumerate}
\item The binary relation $\alpha$ defined on $A$ by
$$(a,a')\in\alpha  \mbox{ if and only if  there exists $b\in B$ such that } (a,b),(a,b')\in C$$
is a congruence of $\mathbb{A}$. The analogous statement is true of the dual relation $\beta$ defined on $B$.
\item If $\mathbb{C}'\unlhd \mathbb{C}\leq_{sp} \mathbb{A}\times\mathbb{B}$ and $\mathbb{C}'\leq_{sp}\mathbb{A}\times\mathbb{B}$, if $\alpha'$ and $\beta'$ are the pair of congruences defined on $A'$ and $B'$, respectively, as in (1), then $\alpha=\alpha'$ and $\beta=\beta'$.
\end{enumerate}
\end{proposition}

We will refer to the congruences $\alpha$ and $\beta$, defined as in Part (1) of the Proposition \ref{absorb} , as the \emph{linkedness} congruences on $\mathbb{A}$ and $\mathbb{B}$ induced by $\mathbb{C}$. We say that $\mathbb{A}$ and $\mathbb{B}$ are \emph{linked} if $\alpha=1_\mathbb{A}$ and $\beta=1_\mathbb{B}$ or, equivalently, if $\pi_1\vee\pi_2=1_\mathbb{C}$. If $\alpha=0_\mathbb{A}$ and $\beta=0_\mathbb{B}$, the subdirect product is the graph of an isomorphism between the algebras $\mathbb{A}$ and $\mathbb{B}$.

For Taylor algebras, linked subdirect products satisfy the following property:

\begin{theorem} (L. Barto, M. Kozik, \cite{b-k2}) \label{AbsThm} Let $\mathbb{C}\leq_{sp} \mathbb{A}\times\mathbb{B}$ be a Taylor algebra. If $\mathbb{C}$ is linked then
\begin{itemize}
\item $\mathbb{C}=\mathbb{A}\times\mathbb{B}$, or
\item $\mathbb{A}$ has a proper absorbing subalgebra, or
\item $\mathbb{B}$ has a proper absorbing subalgebra.
\end{itemize}
\end{theorem}

\subsection{Simple Idempotent Algebras}

Let $\mathbb{A}$ be an algebra. We say that $0\in A$ is an \emph{absorbing element} for $\mathbb{A}$ if, for every $(k+1)$-ary term operation $t(x,\bar{y})$, such that $t^\mathbb{A}$ depends on the variable $x$, the following holds for every $\bar{a}\in A^k$:
$$t^\mathbb{A}(0,\bar{a})=0.$$
We remark here that the property of being an absorbing element is stronger than the requirement that $\{0\}$ be an absorbing subuniverse of $\mathbb{A}$.

Given any finite power of an algebra $\mathbb{A}$, say $\mathbb{A}^n$, for $n\geq 2$, and any $n$ congruences $\theta_1,\theta_2,\ldots,\theta_n\in Con(\mathbb{A})$, the binary relation defined on $A^n$ by
$$((a_1,a_2,\ldots,a_n),(b_1,b_2,\ldots,b_n))\in \theta_1\times\theta_2\times\ldots\times\theta_n$$
if and only if $(a_i,b_i)\in \theta_i$, for all $i=1,\ldots,n$, is a congruence on $\mathbb{A}^n$. Therefore,
$$Con(\mathbb{A}_1)\times Con(\mathbb{A}_2)\times\ldots\times Con(\mathbb{A}_n)\subseteq Con(\mathbb{A}^n).$$
We say that a simple algebra $\mathbb{A}$ is \emph{congruence skew-free} if the equality holds, i.e. if 
$$Con(\mathbb{A}^n)\cong \mathbf{2}^n,$$
for every $n\geq 1$, where $\mathbf{2}$ is a two-element lattice.

The crux of our proof of the Dichotomy Conjecture lies in the analysis of subdirect products of simple absorption-free idempotent algebras. The following theorem provides the key to understaning the aforementioned subdirect products:

\begin{theorem} (K. Kearnes, \cite{Kearnes}) If $\mathbb{A}$ is an idempotent simple algebra, then exactly one of the following conditions is true:

\begin{enumerate}
\item $\mathbb{A}$ has a unique absorbing element.
\item $\mathbb{A}$ is Abelian.
\item $\mathbb{A}$ is congruence skew-free.
\end{enumerate}
\end{theorem}

In fact, more can be said of $\mathbb{A}$, if $\mathbb{A}$ is Abelian. The following theorem provides a much tighter structural characterization in that case:

\begin{theorem} (M. Valeriote, \cite{Valeriote}) Every simple Abelian  algebra is strictly simple, i.e. it contains no proper nontrivial subalgebras.
\end{theorem}

In fact, there is a very precise characterization of strictly simple idempotent Abelian algebras (see e.g. \cite{szendrei}):

A finite idempotent Abelian algebra $\mathbb{A}$ is strictly simple if and only if there exist a finite field $K$ and a finite-dimensional vector space $V$ over $K$ such that $\mathbb{A}$ is term equivalent to the algebra
$$(V; x-y+z, \{\lambda x+ (1_K-\lambda )y \, \vert  \, \lambda \in K\})$$
where $+$ is the addition of vectors, $1_K$ is the multiplicative identity of the field $K$, and $\lambda x$ is the scalar multiplication by $\lambda \in K$ in $V$. 

The last algebraic fact we will list here is a fact about subdirect products of simple Maltsev algebras. For the proof, see e.g. \cite{Burris1981}

\begin{theorem} \label{maltsev} Let $\mathbb{A}_1,\ldots,\mathbb{A}_n$ be simple algebras in a Maltsev variety. If
$$\mathbb{B}\leq_{sp} \mathbb{A}_1\times\ldots\times\mathbb{A}_n$$
is a subdirect product, then
$$\mathbb{B}\cong \mathbb{A}_{i_1}\times\ldots\times\mathbb{A}_{i_k}$$
for some $\{i_1,\ldots,i_k\}\subseteq \{1,\ldots,n\}$. 

In particular, if $\mathbb{A}$ and $\mathbb{B}$ are two Maltsev algebras then any subdirect product
$$\mathbb{C}\leq_{sp} \mathbb{A}\times\mathbb{B}$$ 
is either the direct product or the graph of an isomorphism $f:\mathbb{A}\rightarrow\mathbb{B}$.
\end{theorem}

\section{Datalog, Linear Arc Consistency, and Singleton Linear Arc Consistency}

A \emph{Datalog program} for a relational template $\mathbf{A}$ is a finite set of rules of the form
$$T_0\leftarrow T_1,T_2,\ldots, T_n$$
where $T_i$'s are atomic formulas. 
$T_0$  is the \emph{head} of the rule, while $T_1,T_2,\ldots, T_n$ form the \emph{body} of the rule.
Each Datalog program consists of two kinds of relational predicates:
the \emph{intentional} ones (IDBs), which are those occurring at least once in the head of some rule and which are not part of the original signature of the template (they are derived by the computation.)
The remaining predicates are said to be the \emph{extensional} ones, or EDBs. They are relations from the signature of the template and do not change during computation; i.e. they cannot appear in the head of any rule.
In addition to those, there is one special, designated IDB, which is nullary (Boolean) and referred to as the \emph{goal} of the program.

We say that the rule
$$T_0\leftarrow T_1,T_2,\ldots, T_n$$
is \emph{linear} if at most one atomic formula in its body is an IDB. A Datalog program is linear if so are all its rules. 

The semantics of Datalog programs are generally defined
in terms of fixed-point operators. We are particularly interested in the Datalog programs which, being presented a relational template $\mathbf{A}$, verify if the template satisfies certain consistency requirements in terms of witnessing path patterns prescribed by the CSP instance in question. 

\subsection{Linear Arc Consistency}

Given a CSP instance $\mathcal{I}$ over a relational template $\mathbf{A}$, a Datalog program verifying its linear arc consistency has one IDB $B(x)$, for each subset $B\subseteq A$ in the instance. To construct rules for the program, we consider a single constraint $R(x_{i_1},x_{i_2},\ldots, x_{i_m})$, with $R$ being a $k$-ary relation in the signature of $\mathbf{A}$, and two variables $x_{i_j},x_{i_k}$ in its scope. If a fact $B(x_{i_j})$ has already been established about $x_{i_j}$, we add the rule
$$C(x_{i_k}) \leftarrow R(x_{i_1},x_{i_2},\ldots, x_{i_m}), B(x_{i_j}).$$
The collection of all such rules, along with the goal, is said to be a Datalog program verifying the \emph{linear arc consistency} of the instance. If the goal predicate is derived, the instance is not linearly arc consistent; otherwise, we say that it  has linear arc consistency, or LAC, for short.

The complexity of verifying LAC for an instance is in nondeterministic log-space, since it reduces to verifying reachibility in a directed graph.

\subsection{Singleton Linear Arc Consistency}\label{SLAC}

Singleton linear arc consistency (or, SLAC, for short) is a consistency notion provably stronger than linear arc consistency. A recent result of M. Kozik (\cite{Kozik2016}) proves that, in fact, all CSPs over the templates of bounded width can be solved by SLAC, whereas, under the assumption that \textsc{NL}$\neq$ \textsc{P}, there are CSPs over the bounded width templates which cannot be solved by LAC, for instance \textsc{3-HORN-SAT}, the satisfiability of Horn formulas in the 3-CNF.

We describe the algorithm for verifying SLAC in its procedural form. Given an instance $\mathcal{I}$, we introduce a unary constraint $B_x$, for each variable in the instance and update them by running the LAC algorithm with the value of $x$ being fixed to an arbitrary $a\in B_x$.

\begin{algorithm}[H]
   \caption{SLAC Algorithm}
    \begin{algorithmic}[1]
     
        \For{ every variable $x$ of $\mathcal{I}$}
            \State Introduce the unary constraint $B_x :=(x,A)$
        \EndFor
        \Repeat
                \For{ every variable $x$ of $\mathcal{I}$}  
                    \State $C:=A$
                    \For{ every value $a\in A$}
                          \State run LAC on the restriction of $\mathcal{I}$ with $B_x=\{a\}$ and constraints modified accordingly
                          \If { LAC results in contradiction}
                                \State remove $a$ from $C$
                         \EndIf
                    \EndFor
                    \State $B_x:=(x,C)$
                 \EndFor
        \Until{ There are no further changes in $B_x$}
\end{algorithmic}
\end{algorithm}

In this paper, we will be using the multisorted version of SLAC. What we mean by that, is that the predicates for the domains of different variables $x$ are assumed to be the subsets of different sorted domains, generated by the reduction to a binary instance. Since the domains produced by the reduction to the binary case are positive-primitive definable, this presents no particular issue.

\section{Patterns and steps}

We will create SLAC instances of structures with binary constraints, and, to that end, we define the notions of a pattern and a step. Our definitions will be special cases of the more general ones given in \cite{Kozik2016}. We fix an instance $\mathcal{I}$ of a CSP, all of whose constraint relations are binary.

\begin{definition} A \emph{step}  in an instance $\mathcal{I}$ is a pair of variables which is the scope of a constraint in $\mathcal{I}$. A \emph{path-pattern} from $x$ to $y$ in $\mathcal{I}$ is a sequence of steps such that every two steps correspond to distinct binary constraints and which identifies each step's end variable with the next step's start variable. A \emph{subpattern} of a path-pattern is a path-pattern defined by a substring of the sequence of steps. We say that a path-pattern is a \emph{cycle} based  at $x$ if both its start and end variable are $x$.
\end{definition}

\begin{definition} Let
$$p=(x_1,x_2,\ldots,x_k)$$
be a path-pattern. A \emph{realization} of $p$ is a $k$-tuple $(a_1,\ldots,a_k)\in \mathbb{S}_{x_1}\times\ldots\times \mathbb{S}_{x_k}$ such that $(a_i,a_j)$ satisfies the binary constraint associated with the $(x_i,x_j)$-step. If $p$ is a path-pattern with the start variable $x_i$ and $A\subseteq S_{x_i}$, we denote $A+p$ the set of the end elements of all realizations of $p$ whose first element is in $A$. $-p$ will denote the inverse pattern of $p$, i.e. the pattern obtained by reversing the traversal of the pattern $p$. In that case, we define $A-p = A+(-p)$.
\end{definition}

We also make the following observations:

\begin{enumerate} 
\item The LAC algorithm does not derive a contradiction on the instance $\mathcal{I}$ if and only if every path-pattern in $\mathcal{I}$ has a solution.
\item If an instance $\mathcal{I}$ is a SLAC instance then, for every variable $x$ and every $a\in \mathbb{S}_x$, and every path pattern $p$ which is a cycle based at $x$, there exists a realization of $p$ with $x$ being assigned the value  $a$. 
\end{enumerate}

\section{Reduction to binary relations}\label{Binary}

In this section, we outline the reduction of an arbitrary instance with a sufficient degree of consistency to a binary one. The construction is due to L. Barto and M. Kozik and we largely adhere to their exposition in \cite{b-k1}.

An instance is said to be \emph{syntactically simple} if it satisfies the following conditions:

\begin{itemize}
\item every constraint is binary and it its scope is a pair of distinct variables $(x,y)$.
\item for every pair of distinct variables $x,y$, there is at most one costraint $R_{x,y}$ with the scope $(x,y)$.
\item if $(x,y)$ is the scope of $R_{x,y}$, then $(y,x)$ is the scope of the constraint $R_{y,x}=\{(b,a) \, \vert\, (a,b)\in R_{x,y}\}$ (\emph{symmetry of constraints}).
\end{itemize}

Given the Taylor algebra $\mathbb{A}$ parametrizing the instance $\mathcal{I}$, such that the maximal arity of a relation in $\mathcal{I}$ is $K$, we run the algorithm verifying the $(2\lceil\frac{K}{2}\rceil,3\lceil\frac{K}{2}\rceil)$-consistency on $\mathcal{I}$. If the algorithm terminates in failure, we output ``$\mathcal{I}$ has no solution." If the algorithm terminates successfully, we output a new, syntactically simple instance $\mathcal{I}'$ in the following way:

\begin{itemize} 
\item The instance is parametrized by $\mathbb{A}^{\lceil \frac{K}{2}\rceil}$, which is a Taylor algebra. Since $\mathbb{A}$ generates a Taylor variety, which has a weak near unanimity term and, then, so does the variety generated by $\mathbb{A}^{\lceil\frac{K}{2}\rceil}$.
\item For every $\lceil\frac{K}{2}\rceil$-tuple of variables in $\mathcal{I}$, we introduce a new variable in $\mathcal{I}'$ and, if $x=(x_1,\ldots,x_{\lceil\frac{K}{2}\rceil})$ and $y=(y_1,\ldots,y_{\lceil\frac{K}{2}\rceil})$ with $x\neq y$, we introduce a constraint 
\begin{multline*}
R_{x,y}=\{((a_1,\ldots,a_{\lceil\frac{K}{2}\rceil}),(b_1,\ldots,b_{\lceil\frac{K}{2}\rceil}))\,\vert \\ (a_1,\ldots,a_{\lceil\frac{K}{2}\rceil},b_1,\ldots,b_{\lceil\frac{K}{2}\rceil}) \mbox{ admit a consistent $K$-assignment of values }\}.
\end{multline*}
\end{itemize}

The binary instance $I'$ constructed in this way will have a solution if, and only if, the instance $I$ has a solution. 

\begin{definition} Let $l\geq k>0$ be two integers. We say that a CSP instance $\mathcal{I}$ is $(k,l)$-\emph{minimal} if:

\begin{enumerate}
\item Every tuple of distinct variables of length at most $l$ is the scope of some constraint of $\mathcal{I}$.
\item For every $k$-tuple $\bar{x}$ of distinct variables, and every pair of constraints $C_1$ and $C_2$ of $\mathcal{I}$ whose scopes contain $\bar{x}$ among its variables, the projections of $C_1$ and $C_2$ to the variables $\bar{x}$ coincide.
\end{enumerate}
\end{definition}

\section{Cyclic Constraint Satisfaction Problems}

In this section, we investigate a rather specific type of the constraint satisfaction problem, which will play the crucial role in defining the consistency notion needed in the remainder of the paper. 

A \emph{cyclic} CSP (or, CCSP, for short) is a constraint satisfaction problem which has as its domains isomorphic simple absorption-free Taylor algebras, and all of whose constraints are binary, and which is 1-consistent. From the discussion in Section \ref{absorption}, we know that each constraint relation between two domains $\mathbb{S}_x$ and $\mathbb{S}_y$ is either the graph of an isomorphism or a full direct product $\mathbb{S}_x\times \mathbb{S}_y$. 

The classification of finite simple idempotent algebras, which are absorption-free, suggests that a CCSP may be one of the following two types:

\begin{enumerate}
\item A system of linear equations in two variables over a finite field; 
\item A binary CSP over a congruence skew-free, absorption-free simple algebra.
\end{enumerate}

For each CCSP $\mathcal{I}$ over a simple absorption-free algebra $\mathbb{A}$, we can define the accompanying undirected \emph{instance graph} $G_{\mathbb{A}}(\mathcal{I})$ in the following way: the vertices of the graph are all domains $\mathbb{S}_x$ of $\mathcal{I}$ and two vertices $\mathbb{S}_x$ and $\mathbb{S}_y$ have an edge between them if, and only if, the binary constraint relation $R_{x,y}$ is the graph of an isomorphism. We can compute the connected components of this graph in logspace, using Reingold's algorithm (\cite{Reingold05}). 

It is not difficult to see that, in order to solve such a CSP, we need to solve it in each connected component of $G_{\mathbb{A}}(\mathcal{I})$. In the case when the domains are isomorphic simple affine modules, this can be accomplished using the familiar Gaussian elimination algorithm. 

The solvability of the CCSP in the case of a simple, congruence skew-free, absorption-free algebra is less obvious. First, one needs to establish the so-called rectangulation property for subdirect products of such algebras.

\begin{proposition} Let $\mathbb{A}_1,\ldots,\mathbb{A}_k$ be isomorphic simple, absorption-free, congruence skew-free algebras lying in a Taylor variety. If $\mathbb{R}\leq_{sp} \prod_i  \mathbb{A}_i$ is such that $\pi_i\vee \pi_j=1_\mathbb{R}$, then $\mathbb{R}=\prod_i \mathbb{A}_i$. In addition, $\mathbb{R}$ is absorption-free.
\end{proposition}

\begin{proof}  We prove both statements simultaneously, by induction on $k$. If $k=2$, the statements follow from Theorem \ref{AbsThm}. Assume $k\geq 3$ and consider $\mathbb{R}$ as a subdirect product of two algebras:
$$\mathbb{R}\leq_{sp} (\mathbb{A}_1\times\ldots\times\mathbb{A}_{k-1})\times\mathbb{A}_k.$$
 Let $\mathbb{A}'$ denote $\mathbb{A}_1\times\ldots\times\mathbb{A}_{k-1}$ and let $\alpha,\beta$ be the linkedness congruences on $\mathbb{A}'$ and $\mathbb{A}_k$, respectively. By inductive hypothesis, $\mathbb{A}'$ is absorption-free which yields two possibilities: either both linkedness congruences $\alpha$ and $\beta$ are full congruences on their respective algebras or $\beta=0_{\mathbb{A}_k}$. If the former is the case, we get the desired conclusion, after another application of Theorem \ref{AbsThm}. We proceed to show that the assumption that $\beta=0_{\mathbb{A}_k}$ leads to a contradiction. Since $\mathbb{A}'/\alpha$ is a simple algebra and all factors are isomorphic and congruence skew-free, 
$$\alpha = \theta_1\times\ldots\times \theta_{k-1},$$
where $\theta_i\in Con(\mathbb{A}_i)$, for $i=1,\ldots,k-1$ and, for precisely one $i$, say $i_0$, 
$$\theta_{i_0}=0_{\mathbb{A}_{i_0}},$$
while, if $1\leq j\leq k-1$ and $j\neq i_0$, $\theta_j=1_{\mathbb{A}_j}$.

However, this violates the assumption that $\pi_{i_0}\vee \pi_k=1_\mathbb{R}$. Therefore,
$$\mathbb{R}=\prod_i \mathbb{A}_i.$$ 

Finally, since both $\mathbb{A}'$ and $\mathbb{A}_k$ are absorption-free and fully linked, their direct product is absorption-free as well.

\end{proof}

As in \cite{Kozik2016}, upon establishing the rectangulation in simple absorption-free congruence skew-free algebras, one can emulate the proof given in that paper to show that, in that case, CCSP will be solvable if, and only if, it is SLAC.

Looking ahead, the algorithm we will construct will be based on pre-processing the instance by enforcing existence of solutions on instances induced by simple absorption-free algebras in $HS(\mathbb{S}_x)$, for all $x\in V$. 

\section{A proof of Main Theorem }

In this section we provide a proof of the Dichotomy Conjecture using the binary instance constructed in Section \ref{Binary}. Therefore, from this point on, we assume that we are working with an instance $\mathcal{I}$ of a CSP, parametrized by a Taylor algebra, which is syntactically simple and binary and (2,3)-minimal.

\subsection{AF-consistency}

The fundamental obstacle in any attempt to directly adapt known algorithms for solving CSPs parametrized by bounded width algebras to the general case of Taylor templates lies in the apparent difficulty to distinguish the computation paths leading to solutions from those leading to failure, based on mere global satisfaction of a local consistency notion in the instance.

In this subsection, we develop the notion of  AF-consistency which can be enforced on a (2,3)-consistent syntactically simple binary instance. 

We define inductively, the AF-consistency checking agorithm $\mathcal{A}_k$, for all CSP instances $\mathcal{I}$ such that $\max_{x}|\mathbb{S}_{x}|\leq k$.

Let $B\leq \mathbb{S}_x$, for $x\in V$. For any $y\in V$, $y\neq x$, we define $R^+_{x,y}(B)=\{c\in \mathbb{S}_y \, \vert \, \exists b\in B, \, (b,c)\in E_{x,y}\}$. It is readily seen that $R^+_{x,y}(B)$ is a subuniverse of $\mathbb{S}_y$.

Next, we define a list of all pairs $(M,\theta_M)$, where $M$ is an absorption-free subuniverse of some $\mathbb{S}_x$, and $\theta_M$ is its maximal congruence.
$$\mathcal{M}= ((M,\theta_M)_i \, : i\in I),$$
so that, if $(M,\theta_M)$ and $(M',\theta_{M'})$ are two elements of the list and $M'$ is contained in a $\theta_M$-block of $M$, then $(M,\theta_M)$ appears in the list $\mathcal{M}$ before $(M',\theta_M')$. The reason for this is the following: if $(M, \theta_M)$ fails the test, there will be no need to examine any subinstances determined by a subuniverse of $M$, so, by removing $M$, we are also removing all of its subuniverses.

We are now ready to state the procedure which enforces AF-consistency

\begin{enumerate}
\item For the next pair $(M,\theta_M)$ in the list $\mathcal{M}$, form the $(M,\theta_M)$-\emph{test instance} in the following way: suppose $M\leq \mathbb{S}_x$, for some $x$. For $y\neq x$, if there exists a congruence $\theta_y$ on $R^+_{x,y}(M)$, such that, if $B_1$ and $B_2$ are two distinct $\theta_M$-blocks and $p$ a path pattern from $x$ to $y$ such that $R^+_{x,y}(M)\cap (B_1+p)$ and $R^+_{x,y}(M)\cap (B_2+p)$ are containt in distinct blocks of $\theta_y$, we will say that the variable $y$ is  \emph{relevant}. Therefore, for each relevant variable $y$,
$$\mathbb{S}/\theta_{y}\cong M/\theta_M.$$
In fact, $\theta_y$ is independent of the choice of the path pattern $p$, because of (2,3)-consistency. 

We define a \emph{strand} to be the set of those congruence blocks in each relevant domain which are linked to the same congruence block of $\theta_M$. The $(M,\theta_M)$-test instance will have as its domains the algebras $R^+_{x,y}(M)/\theta_y$, for $y\neq x$, for all relevant variables $y$.  Since $\mathcal{I}$ is a (2,3)-consistent instance, for any pair of relevant variables $y,z$, distinct from $x$, the binary constraint $E_{y,z}$ induces a subdirect product on $R_{x,y}^+(M)$ and $R_{x,z}^+(M)$, so that the $(M,\theta_M)$-test instance is 1-consistent.

\item The $(M,\theta_M)$-test instance is a CCSP, and using either Gaussian elimination or SLAC, we test whether blocks of $\theta_M$ appear in solutions or not; those which do not are removed. For every solution strand, we test AF-consistency, using $\mathcal{A}_{k-1}$.

\item Enforce (2,3)-consistency.

\item Proceed to the next element in the list $(M',\theta_M')$ which is present in the instance, if there are any left.

\end{enumerate}

There are only polynomially many pairs in the list $\mathcal{M}$ (in fact, $\mathcal{O}(n)$), so the algorithm for enforcing AF-consistency runs in polynomial time. The subinstances which pass the Step 3 of the AF-consistency algorithm will form a list $\mathcal{P}$ and on each subinstance from $\mathcal{P}$, we enforce (2,3)-minimaility independently. The subinstance from Step 3, corresponding to a $\theta_A$ block $B$, will be referred to as the passive subinstance determined by $B$.

\begin{lemma} Let $\mathcal{I}$ be a syntactically simple binary instance and let $\mathcal{I'}$ be the instance produced by applying the AF-consistency algorithm to it. Then, the sets of solutions to $\mathcal{I}$ and $\mathcal{I'}$ coincide.
\end{lemma}

\begin{proof} If there exists a solution $f$ to $\mathcal{I}$ whose projection to the $x$-coordinate is in $M\leq \mathbb{S}_x$, then, its restriction to relevant variables is also a solution of the $(M,\theta_M)$-test instance, viewed as a subinstance of $\mathcal{I}$. If AF-consistency test fails on a $\theta_M$-block, then there cannot be any solutions $f$ projecting into that block in their $x$-coordinate.

Also, the solution projecting into a $\theta_M$-block $B$ in its x-coordinate will lie in its entirety in the subinstance induced by $B$, so this subinstance must not fail the (2,3)-consistency test either.
\end{proof}

In order to clarify the reasons behind introducing the notion of AF-consistency, we remark that, in essence, pre-processing the instance in the described way will effectively remove the branches of the computation tree which lead to failure (i.e. which yield no solutions.) Another way to view this stage of the algorithm has universal algebraic provenance: the obstructions to bounded width are strictly simple algebras of affine type which are in $\textsc{HS}(\mathbb{A})$ (\cite{LVZ}). This pre-processing examines such algebras, among other absorption-free subuniverses, and trivializes them to a single element or removes them altogether in the case when not all strands meet consistency requirements. Otherwise, a strand meeting such a requirement can be chosen arbitrarily, just as in the bounded width algorithms in the literature.

The subinstances determined by the surviving $\theta_M$-blocks, for absorption-free algebras $M$, can be viewed as a collection of polynomially many ``passive'' subinstances of the problem. Implicitly, any reduction via absorption, or otherwise, may be seen as a reduction performed on the passive subinstances. At the point where the transformation in question reduces the problem to one of these subinstances, it becomes active while the subinstances which are not contained in it are discarded by the algorithm. 

\subsection{Reduction to smaller subinstances - outline}

The general idea of the algorithm we are about to present can be described as follows: assuming the variables $V$ of $\mathcal{I}$ have been linearly ordered in some fashion, say $V=\{x_i \; : 1\leq i\leq N\}$, we reduce the domains $\mathbb{S}_{x_i}$ to singletons, so that $\mathbb{S}_{x_i}$ is reduced to a single element before $\mathbb{S}_{x_j}$ is, for $i<j$.

During the reduction of $\mathbb{S}_{x_i}$ to a single element, reductions based on the presence of absorbing subuniverses in $\mathbb{S}_{x_i}$ are used, until $\mathbb{S}_{x_i}$ becomes absorption-free. This is followed by a reduction to a passive subinstance in the list $\mathcal{P}$. These two types of reductions are alternated until, eventually, $\mathbb{S}_{x_i}$ becomes a singleton. During this sequence of reductions, the list $\mathcal{P}$ of passive subinstances is updated and, because of the enforced AF-consistency, never becomes empty.

In what follows, we assume that $\mathcal{I}$ is a 1-consistent SLAC instance which is also AF-consistent, with the accompanying list $\mathcal{P}$ of passive subinstances, such that, for every absorption-free $A\leq \mathbb{S}_{x_j}$, $1\leq j\leq n$, and every maximal congruence $\theta$ of $A$, there is a passive 1-consistent SLAC subinstance generated by every $\theta$-block of $A$.

\subsubsection{Absorption is present in $\mathbb{S}_{x_i}$}

In this case,  some $\mathbb{S}_{x_i}$ contains a proper absorbing subuniverse $B$. The reduction via absorption from Kozik's paper adapts to this setting and can be applied to the instance $\mathcal{I}$. This particular choice of $B$ implicitly defines reductions on all passive subinstances from $\mathcal{P}$. If the reduced subinstance $\mathcal{I}'$ fails to intersect a passive subinstance from $\mathcal{P}$, that passive subinstance is removed from $\mathcal{P}$. For the instance $\mathcal{I}$, if $\mathbb{S}_x$ is not absorption-free, the analysis of the proof in Section 10 of \cite{Kozik2016} indicates that $B\unlhd \mathbb{S}_x$ can always be chosen in such a way that $B$ is a minimal absorbing subuniverse of $\mathbb{S}_x$ and which is, therefore, absorption-free.

It is easily seen that, for any such choice of $B$, $\mathcal{P}$ cannot become empty: namely, by considering any maximal congruence $\psi$ on $B$, we see that the passive subinstances determined by the blocks of $B/\psi$ must remain in $\mathcal{P}$.

The following fact has an obvious proof, based on the definition of an absorbing subuniverse:

\begin{lemma} Let $B\leq A$ and $C\unlhd A$. Then, if $C\cap B\neq\emptyset$, $C\cap B\unlhd B$.
\end{lemma}

\begin{proposition} \label{useful} Let $\mathcal{I}$ be a 1-consistent syntactically simple binary instance with domains $\mathbb{S}_x$, $x\in V$. If $A\unlhd \mathbb{S}_x$, then $R^+_{x,y}(A)\unlhd \mathbb{S}_y$, for all $y\neq x$.
\end{proposition}

Using this fact, we see that, if $M\leq \mathbb{S}_x$ is absorption-free, then every absorbing subuniverse of $C\unlhd \mathbb{S}_x$, such that $C\cap M\neq\emptyset$, must satisfy $M\leq C$. In addition, for every $y$, $R^+_{x,y}(M)\leq R^+_{x,y}(C)$. The proof from \cite{Kozik2016} shows that SLAC remains preserved in all passive subinstances under the absorption reduction defined in that paper, unless a domain of the passive subinstance fails to intersect the minimal absorbing subuniverse in $\mathbb{S}_{x}$ in some coordinate $x$.

\subsection{Absorption is absent from the instance}\label{absfree}

Next, we consider the case when all the domains $\mathbb{S}_x$ in the 1-consistent, AF-consistent and SLAC instance $\mathcal{I}$ have no proper absorbing subuniverses. In addition, all passive subinstances in $\mathcal{P}$ are 1-consistent and SLAC.

Let $\theta$ be a maximal congruence of $\mathbb{S}_x$. We recall that when we defined the $(\mathbb{S}_x,\theta)$-test instance, the \emph{relevant} variables were defined in the following way: if $M\leq \mathbb{S}_x$, for some $1\leq i\leq n$. For $y\neq x$, if there exists a congruence $\theta_y$ on $R_{x,y}(M)$, such that, if $B_1$ and $B_2$ are two distinct $\theta_M$-blocks and $p$ a path pattern from $x$ to $y$ such that $R^+_{x,y}\cap (B_1+p)$ and $R^+_{x,y}\cap (B_2+p)$ are containt in distinct blocks of $\theta_y$, we defined the variable $y$ to be  \emph{relevant}.  In other words, the variable $y$ is non-relevant if, and only if, \emph{every} path-induced subdirect product
$$C\leq_{sp} \mathbb{S}_x/\theta\times R^+_{x,y}(\mathbb{S}_x)$$
is linked. In particular, the subdirect product induced by $E_{x,y}$
$$C\leq_{sp} M/\theta_M\times R^+_{x,y}(M)$$
is linked. A result from \cite{b-k1}, in essence, states the following

\begin{proposition} Let $C\leq_{sp} A\times B$ be a subdirect product of finite Taylor algebras $A$ and $B$, such that $C$ is linked. Let $A_1\unlhd A$ and $B_1\unlhd B$, and $C'=C\cap (A_1\times B_1)$ be a subdirect product of $A_1$ and $B_1$. Then, $C'$ is linked.
\end{proposition}

This, in combination with Proposition \ref{useful} implies that if $y$ was a non-relevant variable in the original $(\mathbb{S}_x,\theta)$-test instance, then
the subdirect product
$$\mathbb{S}_x/\theta\times R^+_{x,y}(\mathbb{S}_x)$$
remains linked under the reductions. Since both $\mathbb{S}_x$ and $\mathbb{S}_y$ are absorption-free, the subdirect product is the full direct product. 

Also, if $A$ is an absorption-free subuniverse of some $\mathbb{S}_{x}$ and $a\in \mathbb{S}_y$, which is contained in the passive subinstance in $\mathcal{P}$ generated by $(B,\theta_B)$, if $a\in \mathbb{S}_y$ is linked to more than one $\theta_B$-block, the variable $y$ was non-relevant in the $(B,\theta_B)$-test  instance, so $a$ must be connected to all the strands of the subinstance and any reductions in the $y$-coordinate do not alter that property.

For that reason, the AF-consistency will be preserved since the cyclic CSP on the strands is the same one as in the test instance. Consequently, we can pick any solution strand and replace $\mathcal{I}$ with a smaller subinstance $\mathcal{I}'$; namely, the passive subinstance from $\mathcal{P}$ generated by the $\mathbb{S}_x/\theta$-block of the strand in question.

\section{Conclusion and future directions}

We have presented an algorithm for solving constraint satisfaction problems over finite templates with Taylor polymorphisms, which is based on consistency checks, which solves binary bounded width problems and a different consistency check, the AF-consistency, which solves localized absorption-free problems appearing in the computation tree and, effectively, removes the unsuccessful branches leading to no solutions. 

The AF-consistency check ensures that the variables can be chosen consistently within each connected component and, in the case when the induced CCSP is over a simple affine module, that the generated system of linear equations over a finite field is solvable. Such systems are called \emph{cyclic systems of equations} and have also been studied from the point of view of finite model theory, since their definability in various expansions of fixed point logic is intimately related to the expressibility of the graph isomorphism problem for CFI graphs (see e.g. \cite{pakusa2015}).

Instead of using Gaussian elimination, such systems can be solved using Reingold's algorithm for reachability in undirected graphs, with more details being given implicitly in \cite{egri2014space}, where this approach is used to develop a logspace algorithm for solving conservative CSPs over the digraphs satisfying the so-called, Hagemann-Mitschke equations in their algebra of polymorphisms. It would be interesting to see if the reduction in the case of simple absorption-free congruence skew-free algebras can be carried out in the same way, i.e. whether this reduction can be carried out in advance, using Reingold's algorithm, instead of establishing the rectangulation theorem and then using SLAC. This would be the case if the following question has the affirmative answer:

\begin{openproblem} If $\mathbb{A}$ is a simple idempotent, absorption-free algebra which is congruence skew-free, does $\mathbb{A}$ satisfy Hagemann-Mitschke identities in its algebra of polymorphisms? In terms of tame congruence theory (for more details, the reader is invited to consult \cite{hobby1988structure}), is it the case that
$$\textsc{typ} (V(\mathbb{A}))\subseteq \{2,3\}?$$
In fact, is $V(\mathbb{A})$ a Maltsev variety?
\end{openproblem}

If this is indeed the case, it would suggest that a sufficiently general algorithm for solving CSPs over Taylor templates can be given, which would eschew finer algebraic analysis of the templates but which would, instead, rely solely on local consistency checks and graph connectivity in the binary case. This would also bring into sharp focus the true reason for tractability for Taylor domains: the validity of the Absorption Theorem, which plays the crucial role in the construction of the algorithm presented here. Namely, essentially unary algebras fail the Absorption Theorem rather miserably since every subuniverse of such an algebra is absorption-free and the Absorption Theorem is rendered meaningless. From the general point of view, the failure of absorption results in the inability to establish any kind of rectangulation between ``localized'' solutions, i.e. the solutions over independent subinstances, which makes such a problem much harder to solve, from a naive standpoint.

Another interesting problem would be to investigate whether this algorithm can be expressed in a logic which is a promising candidate for capturing polynomial time. In order for such an extension of first-order logic to exist, in addition to the expected recursion mechanisms, it must be able to express solvability of systems of linear equations over finite fields, and, even more generally, over finite Abelian groups of the type $\mathbb{Z}_{p^k}$. This has proved to be a nontrivial property of any candidate logic. For a more thorough discussion of these topics, see e.g. \cite{dawar2012definability}, \cite{gradel2015characterising}, and \cite{pakusa2015}.

\begin{openproblem} Can this algorithm be defined in any of the following extensions of first-order logic: \textsc{LFP + Rk}, \textsc{PIL+C}, or \textsc{CPT+C}?
\end{openproblem}

\bibliography{digraph_reduction} 
\bibliographystyle{siam}

\end{document}